\documentclass[letterpaper]{article} 
\usepackage{aaai2026}  
\usepackage{times}  
\usepackage{helvet}  
\usepackage{courier}  
\usepackage[hyphens]{url}  
\usepackage{graphicx} 
\usepackage{natbib}  
\usepackage{caption} 
\usepackage{algorithm}
\usepackage{algorithmic}

\usepackage{newfloat}
\usepackage{listings}
\DeclareCaptionStyle{ruled}{labelfont=normalfont,labelsep=colon,strut=off} 
\floatstyle{ruled}
\newfloat{listing}{tb}{lst}{}
\floatname{listing}{Listing}
\usepackage{algorithm}
\usepackage{algorithmic}
\usepackage{multirow}
\usepackage{makecell}
\usepackage{booktabs}
\usepackage{graphicx}
\usepackage{xcolor}
\usepackage{stmaryrd} 

\usepackage{multirow}
\usepackage{amsmath, amsfonts, amssymb, amsthm}
\usepackage{xspace}
\usepackage{tcolorbox}

\definecolor{mycolor}{rgb}{0.122, 0.435, 0.698}
\newcommand{\result}[1]{%
\begin{tcolorbox}[colframe=mycolor,boxrule=0.5pt,arc=4pt,
      left=5pt,right=5pt,top=3pt,bottom=3pt,boxsep=0pt,width=\columnwidth]%
      {#1}
\end{tcolorbox}%
}

\title{Incoherence as Oracle-less Measure of Error in LLM-Based Code Generation}
\author{
    Thomas Valentin\textsuperscript{\rm 1},
    Ardi Madadi\textsuperscript{\rm 2},
    Gaetano Sapia\textsuperscript{\rm 2},
    Marcel B{\"o}hme\textsuperscript{\rm 2}
}
\affiliations{
    \textsuperscript{\rm 1}ENS Paris-Saclay, France\\
    \textsuperscript{\rm 2}Max Planck Institute for Security and Privacy, Germany\\
    thomas.valentin@ens-paris-saclay.fr, \{ardi.madadi, gaetano.sapia, marcel.boehme\}@mpi-sp.org
}

\newtheorem{theorem}{Theorem}[section]

\newcommand{\SampleSpace}{\Omega}

\newcommand{\Prob}{\mathbb{P}}

\newcommand{\Descr}{\mathsf{Descr}}
\newcommand{\prog}{\pi}
\newcommand{\ProgRV}{\Pi}
\newcommand{\Progs}{\mathsf{Prog}}
\newcommand{\sem}[1]{\llbracket #1 \rrbracket}

\newcommand{\gtfun}[1]{f^{\ast}_{#1}}
\newcommand{\Input}{\mathsf{Input}}
\newcommand{\Output}{\mathsf{Output}}

\DeclareMathOperator{\Err}{\mathcal{E}}
\DeclareMathOperator{\Inc}{\mathcal{I}}
\newcommand{\passone}{\texttt{pass@1}\xspace}

\newcommand{\coder}{\mathsf{Coder}}
\newcommand{\gen}{\mathsf{Gen}}

\newcommand{\toolname}{DiffTrust}

\newcommand{\artifact}{\textcolor{blue}{\url{https://github.com/mpi-softsec/difftrust}}\xspace}

\begin{document}

\maketitle

\begin{abstract}
Generating code from a natural language programming task is one of the most successful applications of Large Language Models (LLMs). Yet, the generated program may be buggy. Without an oracle, such as an existing, correct implementation or a formal specification, can we somehow estimate how likely the generated program is correct?

In this paper, we propose a measure of incorrectness, called \emph{incoherence}, that can be estimated efficiently in the absence of an oracle and allows us to establish a lower bound on the error, i.e., the probability that the LLM-generated program for that specification is incorrect.
In our experiments, our incoherence-based methodology can automatically identify about two-thirds of incorrect programs without reports of false positives for the average task.
In fact, \emph{an oracle-based evaluation of LLMs can be reliably replaced by an incoherence-based evaluation}. In particular, we find a very strong agreement between the ranking of LLMs by the number of programs deemed correct via an oracle (pass@1) and the ranking of LLMs by the number of programs deemed correct via incoherence.
\end{abstract}

%
\begin{links}
    \link{Data\,\&\,analysis}{https://github.com/mpi-softsec/difftrust}
    \link{Extended version}{https://arxiv.org/abs/2507.00057}
\end{links}

\section{Introduction}
\label{section:introduction}
LLMs have demonstrated remarkable performance on code generation tasks. Yet, confabulation remains a key concern. Models often produce syntactically correct but functionally incorrect code, raising the critical question of when such outputs can be trusted.
For instance, \citet{bugsLLM2} found that the vast majority of auto-generated programs for easy to medium LeetCode programming tasks are incorrect and explain that 57\% of those do not even properly implement the task (``algorithmic misalignment'') while another 19\% can only be fixed by changing multiple different code locations (multi-hunk).
\citet{bugsLLM} analyzed code generated in scenarios relevant to high-risk cybersecurity weaknesses and found that 40\% of the 1.7k LLM-generated programs actually contain security vulnerabilities.

While ground truth implementations or regression test suites provide a post-hoc evaluation of the generated code, they are often unavailable in real-world deployments, motivating the need for \emph{correctness proxies}---that is, mechanisms that can flag potential failures without external supervision.

\result{\emph{Can we estimate how likely an LLM-generated program is correct in the absence of an oracle?}}

Our work continues a recent stream of works addressing the confabulation problem using the \emph{disagreement} between independently sampled responses to detect untruthful or erroneous outputs \cite{manakul2023selfcheckgpt,friel2023chainpoll,li2024honestcoder,farquhar2024semantic}. A high disagreement indicates a high factual inconsistency. 

However, existing measures of disagreement provide \emph{no guarantees}; they are fundamentally heuristic in nature. 
Crucially, they can struggle to distinguish between confidently incorrect answers and correct answers generated under uncertainty---especially in complex structured domains like code, where semantics are hard to capture and where correctness is binary and unambiguous.

\begin{figure}[t]
\includegraphics[width=0.49\columnwidth]{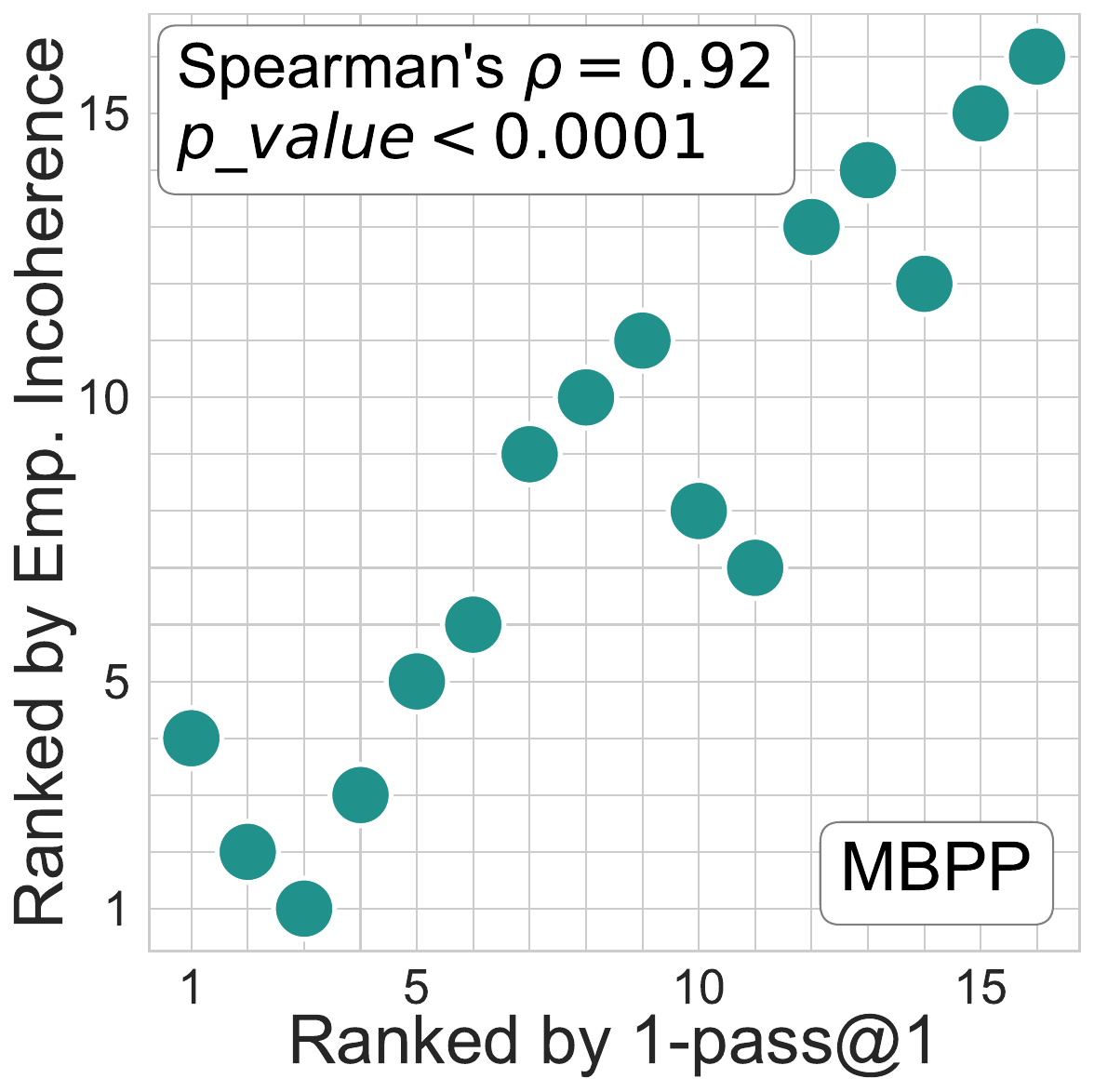}
\includegraphics[width=0.49\columnwidth]{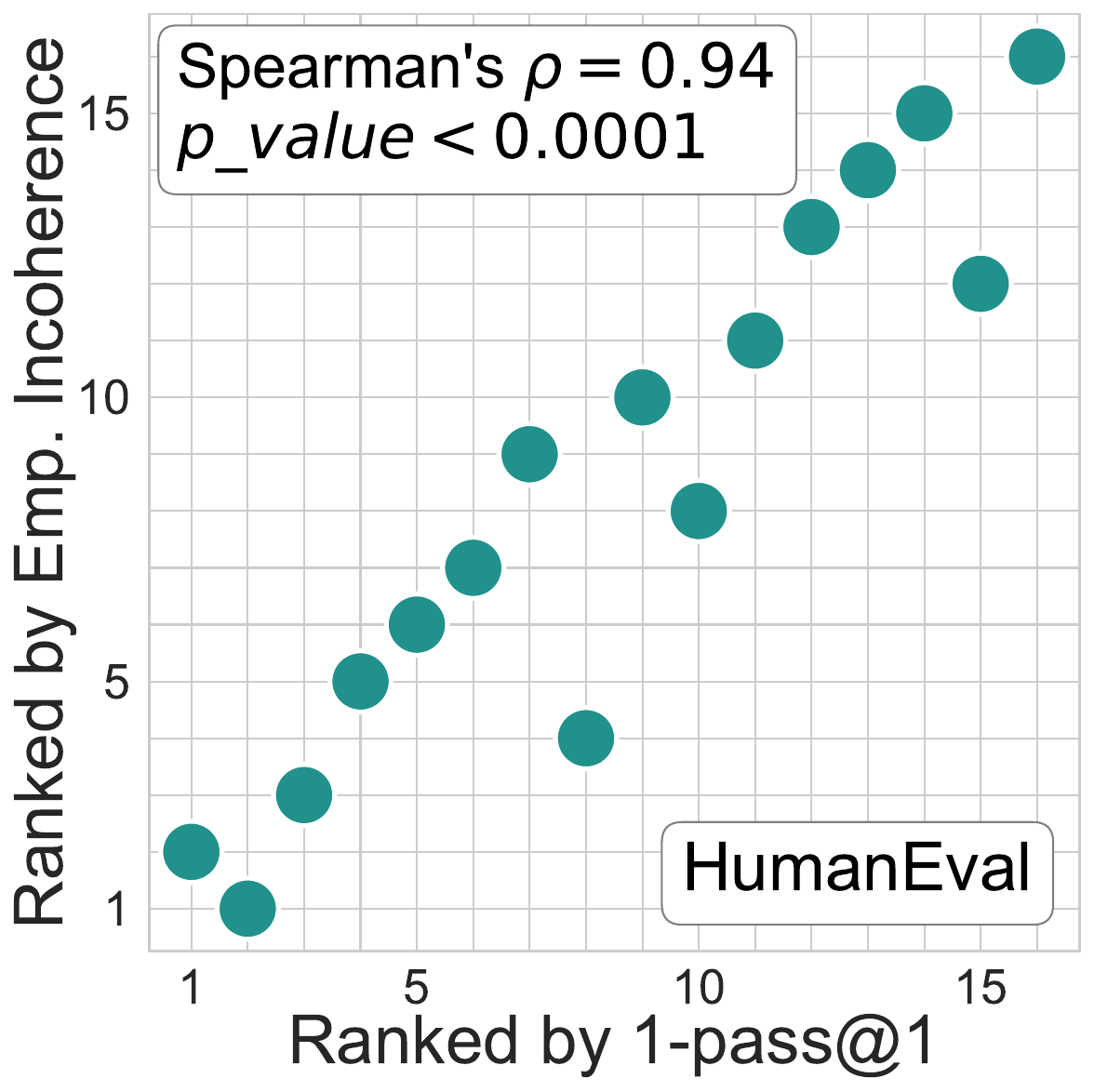}
\caption{Rankings of 16 LLMs on the two most popular code generation benchmarks, MBPP and HumanEval. ``Rank~1" indicates the highest probability of producing correct programs.
\emph{X-axis}: Ranking in terms of $[1-\passone]$ (i.e., the proportion of tasks with non-zero empirical error).
\emph{Y-axis}: Ranking in terms of the proportion of tasks with non-zero empirical incoherence. Note that incoherence can be estimated in the absence of a ground truth implementation.\vspace{-0.2cm}}
\label{fig:mainresult}
\end{figure}

Our key insight is that---in the domain of code---the disagreement between independently sampled solutions for a task can be interpreted \emph{semantically}: if two LLM-generated programs behave differently on the same input, at least one must be incorrect.
If the two programs behave identically across a representative input distribution, we gain empirical confidence in their correctness. This enables a shift from heuristic proxies to semantically grounded ones.

In this work, we formalize this intuition. We argue that incoherence, i.e., the behavioral divergence across samples, is a principled and theoretically justified proxy for model error. Concretely, given an LLM and a programming task $d$, we call the probability that any two programs generated to implement $d$ are functionally different as the LLM's \emph{incoherence} on $d$. If we are also given a ground truth implementation $\gtfun{d}$ for $d$, we call the probability that $\gtfun{d}$ and a program generated to implement $d$ are functionally different as the LLM's \emph{error} on $d$. 
We note that the widely-used \passone score \cite{chen2021evaluating,EvalPlus} on a benchmark set can be written in terms of the empirical error (i.e., the error's maximum likelihood estimate) on each task.

We develop a probabilistic framework that establishes a lower bound on the model's error in terms of incoherence. In contrast to prior work that relies on shallow patterns or internal model metrics, our approach directly leverages the one setting where semantic equivalence is exactly observable: executable code.

In experiments with 16 state-of-the-art LLMs and two popular code generation benchmarks, our incoherence measure, which requires \emph{no} ground truth implementation, works incredibly well as a substitute for \passone. Figure~\ref{fig:mainresult} shows rankings of those LLMs, both in terms of the proportion of tasks with non-zero empirical error (i.e., $[1-\passone]$) and the proportion of tasks with non-zero empirical incoherence. These rankings very strongly agree despite the absence of oracles for our incoherence measure ($\rho\ge 0.92$).
We also find that a non-zero incoherence effectively detects about two-thirds of the non-zero errors in the absence of a ground truth implementation (69\% and 66\% detection rate on MBPP and HumanEval, resp.). No false positives. In cases where the incoherence is zero, the mean error is substantially lower than the average. If we increase the number of generated programs for a programming task 5-fold (from 10 to 50), the detection rate further increases by eight (8) percentage points---of course, at the cost of a 5-fold increase in monetary expenses for additional queries to the LLM. 

\vspace{0.1cm}
\noindent
\textbf{Contributions:}
\begin{enumerate}
\item We propose \emph{incoherence}, a formal, unsupervised proxy for correctness that lower-bounds LLM error.
\item We develop a probabilistic framework linking incoherence to pass@1, with PAC-style efficiency guarantees.
\item Our study shows, incoherence detects errors reliably and yields rankings that agree with oracle-based evaluations.
\item We release all code, results, and analysis scripts.
\end{enumerate}

\section{Background}
\label{section:background}

Code generation is a primary application of LLMs in software engineering \cite{llm4SEsurvey}. Tools like Copilot have seen widespread adoption, with over 40 million installations. Today, even general-purpose LLMs are competitive coders \cite{leaderboards}. In June 2025, nearly 25\% of the 7.5 trillion tokens on~\cite{openrouter} were coding-related.

\citet{llm4codeSurvey} highlight that \emph{trustworthiness} is essential for LLM adoption in programming. Yet, LLMs are prone to generating incorrect code. How can correctness be validated without ground truth? Existing work addresses factuality in text by sampling multiple responses and measuring internal consistency—e.g., SelfCheckGPT~\cite{manakul2023selfcheckgpt}, ChainPoll~\cite{friel2023chainpoll}, and semantic entropy approaches~\cite{farquhar2024semantic}. These methods cluster outputs and estimate confidence based on entropy or token overlap. For code generation, variants include HonestCoder~\cite{li2024honestcoder}, which uses syntax and data-flow modalities, and functional equivalence methods using symbolic execution~\cite{sharma2025assessingcorrectnessllmbasedcode}. However, these approaches are fundamentally heuristic and lack formal guarantees. Their reliance on patterns or output similarity do not capture semantic correctness.

\citet{corina} are interested in generating the most likely correct program for a task, and propose to select from a pool of candidates generated by multiple LLMs that program which best aligns with consensus, where consensus is defined both syntactically and semantically.

Our work addresses the \emph{oracle problem} in software testing \cite{oracle}. While runtime crashes are detectable via sanitizers \cite{asan}, functional correctness is domain-specific and often unobservable. We address this problem by viewing the generated program as a random variable and using incoherence as a proxy for correctness.

\section{Defining Error and Incoherence}
\label{section:problem-setting}

We consider the task of automatically generating a program from a natural language specification. Formally, given a textual description $d$ of a programming task, a code generation system $\coder$---treated as a black-box stochastic process---samples a program $\prog\sim \coder(d)$ intended to satisfy the task description $d$. Our objective is to assess the correctness of these programs $\coder(d)$ without supervision, reference solutions, or access to model internals.

\subsection{Notation}
\label{section:problem-setting:notation}

We write $\Prob(\cdot)$ for probability and $\mathbb{E}[\cdot]$ for expectation, leaving the underlying probability space implicit. For any expression $\textit{expr}$ involving random variables, $\Prob(\textit{expr})$ denotes the probability of the corresponding event. We use $\mathbb{I}(\textit{expr})$ for the indicator function, which is $1$ when $\textit{expr}$ holds and $0$ otherwise.

We denote by $\Descr$ the set of textual function descriptions and by $\Progs$ be the set of programs that define a function.
Let $\sem{\cdot}$ denote the operational semantics such that for all $\prog \in \Progs$,  $\sem{\prog}$ represents the function defined by $\prog$. We refer to $\sem{\prog}$ as the functional interpretation of $\prog$. 

\subsection{Error of a Code Generation System}
\label{section:problem-setting:error-code-generation-system}

We model a code generation system $\coder$ as a function that maps each task $d\in \Descr$ to a corresponding (unknown) distribution over $\Progs$. Formally, for all $d\in \Descr$ : 
\begin{align}
\coder(d) & : \prog \in \Progs \mapsto p_{\prog}^d \in [0,1]
\end{align}%
where $p_{\prog}^d$ is the probability of obtaining $\prog$ when querying $\coder$ with task $d$. A program sampled from $\coder$ for the task $d$ is thus modelled by a random variable that follows the $\coder(d)$ distribution:\vspace{-0.2cm}
\begin{align}
\ProgRV^d \sim \coder(d).
\end{align}

For every task description $d \in \Descr$, we assume there exist an input set $\Input_d$, an output set $\Output_d$ and a correct (deterministic) ground truth implementation $\prog^{\ast}_d \in \Progs$ with its functional interpretation $\gtfun{d} := \sem{\prog^{\ast}_d}$ such that $\gtfun{d}: \Input_d \to \Output_d$.

The \passone score \cite{chen2021evaluating} is a standard metric to evaluate the performance of $\coder$.
For a finite set of tasks $S \subset \Descr$, \passone is defined as the expected fraction of sampled programs that are functionally equivalent to the ground truth implementation:\vspace{-0.1cm}
\begin{align}\label{eq:passone}
\text{\passone}(S) := \mathbb{E}\left[\frac{1}{|S|} \sum_{d \in S} \mathbb{I}\bigl(\sem{\ProgRV^d} = \gtfun{d}\bigr)\right].
\end{align}

We define the \textbf{functional error} of $\coder$ on task $d$ as the complement of \passone computed for a single task $d$, i.e., the probability that the generated program is not functionally equivalent to the ground truth:\vspace{-0.1cm}
\begin{equation}\label{eqn:one}
\Err(d) := \Prob\left(\sem{\ProgRV^d} \neq \gtfun{d}\right) = 1 - \text{\passone}(\{d\}).
\end{equation}
This definition captures the natural notion of error.

Moreover, we introduce a \emph{probabilistic interpretation of correctness} with respect to (w.r.t.) a distribution of inputs. Rather than asking whether the generated function is correct for \emph{all} inputs, which is undecidable due to Rice's theorem, we ask whether it is correct for a \emph{typical input}, drawn from a distribution that represents expected usage. We use this probabilistic interpretation of correctness to introduce a pointwise notion of the error such that \emph{a non-zero pointwise error implies a non-zero functional error} (cf. Eq.~(\ref{eqn:one})).

We model an input generation system $\gen$ as a function that maps each task $d\in \Descr$ to an (unknown) probability distribution  over the corresponding input set $\Input_d$. This distribution might represent how the program is executed under a typical workload. Formally, for all $d\in \Descr$ :
\begin{align}
\gen(d) & : x \in \Input_d \mapsto p_{x}^d \in [0,1]
\end{align}%
where $p_{x}^d$ intuitively models how likely is a function for $d$ to be called on input $x\in \Input_d$.

We define the \textbf{pointwise error} of $\coder$ w.r.t. $\gen$ for any task $d\in \Descr$ as
\begin{equation}\label{eqn:pntwiseError}
\Err_\gen(d) := \Prob(\sem{\ProgRV^d}(X) \neq \gtfun{d}(X))
\end{equation}%
where $X \sim \gen(d)$.
 
While the functional error can be computed only by verification of functional equivalence (an undecidable problem), the pointwise error can be estimated efficiently (c.f. Appendix~\ref{appendix:estimation-inc-err:monte-carlo-estimation-err})---\emph{in the presence of the oracle} $\gtfun{d}$.

We note that a non-zero pointwise error implies a non-zero functional error, i.e.,
\begin{align}\label{eqn:seven}
\left(\Err_\gen(d) > 0\right) \implies (\Err(d) > 0).
\end{align}

Our pointwise error $\Err_\gen(d)$ models the practical reality that a program might be correct on almost all inputs that are empirically observed when the program is tested, deployed, or used in practice.
By evaluating the probability of failure on a representative input distribution, the pointwise error provides a meaningful and practical estimate of the model's reliability in practical scenarios.
The pointwise error also formalizes the experimental setup originally proposed and now widely used to estimate \passone (i.e., the complement of the mean functional error on a fixed set of programming tasks) using a fixed set of random test cases \cite{chen2021evaluating,EvalPlus}.

\section{Incoherence of a Code Generation System}
\label{section:incoherence-code-generation-system}

Our core challenge is to estimate the pointwise error $\Err_\gen(d)$ \emph{in the absence of the oracle} $\gtfun{d}$, i.e., without supervision. We aim to achieve this using only observations from sampled implementations, without relying on any internal details of $\coder$.
To this end, we specialize the disagreement-based hallucination detection approach \cite{manakul2023selfcheckgpt} to the domain of code generation. The precise definition of the (probabilistic) correctness of a program w.r.t. an oracle (i.e., a ground truth implementation) provides us with the unique opportunity to formalize the approach and to introduce actual \emph{probabilistic guarantees}.

We define the \textbf{pointwise incoherence} of $\coder$ w.r.t. an input generation system $\gen$ and a task $d$ as the probability that two independently sampled programs produce different outputs on a generated input, i.e.,\vspace{-0.1cm}
\begin{equation}\label{eqn:pntwiseIncoherence}
\Inc_\gen(d) := \Prob\left(\sem{\ProgRV^d_1}(X) \neq \sem{\ProgRV^d_2}(X)\right)\\[-0.2cm]
\end{equation}
where  $\ProgRV^d_1, \ProgRV^d_2 \overset{iid}{\sim} \coder(d)$ are two independently sampled programs and $X \sim \gen(d)$ is an input sampled from $\gen$ for task $d$.
This quantity captures the model’s internal uncertainty as revealed through behavioral divergence. Crucially, $\Inc_\gen(d)$ is fully observable and efficient to estimate without an oracle (see Appendix~\ref{appendix:estimation-inc-err:monte-carlo-estimation-inc}).

In the following, we show that this notion of pointwise incoherence provides a rigorous lower bound on the pointwise error and that it can be efficiently estimated.
In Appendix~\ref{appendix:functional-incoherence}, we develop the notion of functional incoherence and establish a lower bound on the functional error in terms of the functional incoherence in parallel.

\subsection{\hspace{-0.1cm}Lower\,Bound\,on\,Error\,in\,Terms\,of\,Incoherence}
\label{section:incoherence-code-generation-system:key-theoretical-result}
The pointwise incoherence provides a lower bound on the pointwise error. Intuitively, if two programs disagree on an input, at least one must be wrong; therefore, the probability of disagreement places a floor on the probability of failure.

\begin{theorem}[Pointwise Incoherence Inequality]
\label{theorem:incoherence-inequality}
\[
\forall \gen, \forall d \in \Descr,\quad \Inc_\gen(d) \leq 2 \times \Err_\gen(d).
\]\vspace{-0.6cm}
\end{theorem}\vspace{-0.1cm}

\begin{proof}
See Appendix~\ref{appendix:proof-pointwise-incoherence-ineq}.
\end{proof}\vspace{-0.1cm}
This result establishes $\Inc_\gen(d)$ as a sound and theoretically grounded proxy for estimating model error on $d$. Unlike heuristic confidence or divergence metrics based on representation-level similarity, $\Inc_\gen(d)$ directly, precisely, and formally captures observable functional disagreement. 

Crucially, an error detection method based on our incoherence metric \emph{never produces false positives}. The inequality guarantees that if the model has zero pointwise error on a task—i.e., $\Err_\gen(d) = 0$—then its pointwise incoherence must also be zero: $\Inc_\gen(d) = 0$. This property distinguishes it from all previously proposed unsupervised proxies, which may still flag ``uncertainty'' even when outputs are correct.

\subsection{Incoherence is Efficiently Estimated}
\label{sec:samplecomplexity}

A key advantage of incoherence as a surrogate for correctness is that it can be estimated efficiently and without access to ground-truth implementations. In this section, we formalize this claim by showing that both the pointwise incoherence and the decision problem of detecting non-zero incoherence admit simple, sample-efficient Monte Carlo estimators with standard PAC-style guarantees.

\begin{theorem}[PAC Estimation]\label{thm:estimationIncoherence}
There exists a randomized algorithm that, given parameters $\delta>0$, $\epsilon>0$, code generator $\coder$, input generator $\gen$, and task  $d\in D$, computes $\bar\Inc_\gen(d)$ such that $\Prob(|\bar\Inc_\gen(d)-\Inc_\gen(d)|\le \epsilon)\ge 1-\delta$ using at most $\left\lceil \frac{\log(2/\delta)}{2\epsilon^2} \right\rceil$ samples.
\end{theorem}
A similar theorem for the pointwise error, the randomized algorithms using Monte Carlo estimation, and the proofs for both theorems using a trivial application of Hoeffdings inequality are postponed to Appendix~\ref{appendix:estimation-inc-err:monte-carlo-estimation-err}\&~\ref{appendix:estimation-inc-err:monte-carlo-estimation-inc}.

If we are only interested in the decision problem using a boolean interpretation of correctness, a detection method offers a statistically sound and substantially more sample-efficient means to certify that $\coder$ generates correct programs for a task $d$ w.r.t. a well-specified usage distribution.

\begin{theorem}[PAC Detection]\label{thm:pacDetection}
There exists a randomized algorithm that, given parameters $\delta > 0$, $\epsilon > 0$, code generator $\coder$, input generator $\gen$, and task $d \in \Descr$, returns \texttt{true} if a disagreement is observed and \texttt{false} otherwise, such that:
\begin{itemize}
    \item If the algorithm returns \texttt{true} :  $\Inc_\gen(d) > 0$.
    \item If the algorithm returns \texttt{false} : $\Inc_\gen(d) \leq \epsilon$ with probability at least $1 - \delta$,
\end{itemize}
using at most $\left\lceil \frac{\log(\delta)}{\log(1 - \epsilon)} \right\rceil$ samples.
\end{theorem}

\noindent
A randomized algorithm based on Monte Carlo estimation and the proof is provided in Appendix~\ref{appendix:estimation-inc-err:pac-detection}.

\paragraph{Implication for Error Detection.}
Although the algorithm described in Theorem~\ref{thm:pacDetection} is designed to detect \emph{non-zero incoherence}, we can use it to infer the presence of \emph{non-zero error} due to the theoretical bound established in Theorem~\ref{theorem:incoherence-inequality}, which states:
\[
\Inc_\gen(d) \le 2 \cdot \Err_\gen(d).
\]
This implies that any task $d$ for which $\Inc_\gen(d) > 0$ must satisfy:
\[
\Err_\gen(d) > 0.
\]
Therefore, when the PAC detection algorithm returns \texttt{true} with high probability, we can conclude that the error rate is also bounded away from zero. This provides a conservative but sound certificate of model error without requiring access to a reference implementation.

\section{Practical Considerations}
\subsection{Fixed Sampling Budget for $\coder(d)$}
\label{sec:fixedbudget}

In theory, pointwise incoherence can be estimated efficiently. The estimator defined by Equation~\eqref{eqn:pntwiseIncoherence} is easy to implement, parallelizable, and statistically robust. Voth incoherence estimation and detection admit PAC guarantees with low sample complexity (cf. Thm.~\ref{thm:estimationIncoherence}, Thm.~\ref{thm:pacDetection}).

In practice, however, the primary bottleneck in large-scale evaluation is not sampling from the input distribution $\gen$, which is typically inexpensive, but generating programs from $\coder(d)$, which typically requires querying an LLM. This cost can be substantial, especially when applied across a large set of tasks $d \in \Descr$.

To reduce this cost, we adopt a fixed sampling budget strategy: given a budget $m$, we draw $m$ programs $\Progs_m = \langle \prog_1, \ldots, \prog_m \rangle$ once from $\coder(d)$ and define an empirical code generator $\coder_m(d)$ as the uniform distribution over these programs:
\[
\coder_m(d) := \text{Uniform}(\Progs_m).
\]
This empirical generator approximates the original distribution $\coder(d)$ while avoiding repeated expensive LLM queries at test time. As $m$ increases, $\coder_m(d)$ converges to $\coder(d)$ in distribution, and the resulting estimates of incoherence become more faithful.

While this approximation introduces some additional variance, it is highly effective in practice. It amortizes LLM sampling costs across many evaluations, enabling scalable incoherence and error estimation. Experimentally, we find that a larger $m$ consistently yields more reliable estimates.

\subsection{Test Input Generation to Implement $\gen$}
\label{sec:inputgeneration}
Automatic software test input generation is a well-studied problem in the software engineering community.
Cast as a \emph{constraint satisfaction problem}, we can use symbolic execution to generate inputs that exercise the different paths of a program~\cite{king} or that reveal a difference between two program versions~\cite{prv}.
Cast as an \emph{optimization problem}, we can use heuristic search to generate inputs that maximize code coverage \cite{korel}.

For our purposes, we propose to use \emph{fuzzing}, an approach that mutates a set of user-provided or auto-generated seed inputs to generate new inputs. Today, fuzzing is the most successful and most widely-deployed automatic testing technique in practice \cite{fuzzing}. Like random test input generation, fuzzing is amenable to \emph{statistical guarantees}, e.g., to quantify the probability of finding a bug with the next generated input in an ongoing testing campaign that has found no bugs \cite{assurance,stads,residualrisk,ICSE26-dependency,ICLR25-unseen,reachability}.

In fact, fuzzing has recently been proposed specifically to improve the soundness of the evaluation of LLM-based code generators on the HumanEval and MBPP benchmarks \cite{EvalPlus}, where \passone (i.e., mean error across all benchmark tasks) was traditionally computed using five test inputs per task \cite{chen2021evaluating}. The technique EvalPlus constructs the input distributions $\gen$ in two stages:
\begin{enumerate}
    \item \textbf{Seed Corpus Generation}: An LLM is prompted with the specification (or the ground truth implementation) to produce a set of canonical input examples.
    \item \textbf{Type-Aware Mutation}: These examples are mutated using transformations that preserve the input types but introduce variation (e.g., altering values, shuffling list contents, varying string formats).
\end{enumerate}

We observe that Stage~1 might introduce a bias where an LLM's generated code might appear to perform better on inputs generated by the same LLM, compared to inputs generated by another LLM. Hence, in our experiments, to provide a fair evaluation of all considered LLMs, we mitigate that potential bias by using the benchmark-provided test inputs as seed inputs. In practice, in the absence of existing test inputs, we suggest using the original method or discovering the seed corpus using greybox fuzzings \cite{thinair}.

\section{Experimental Setup}
\label{section:experimental-setup}
\subsection{Research Questions}
\label{section:experimental-setup:RQ}
Our study aims to answer the following research questions.
\begin{itemize}
  \item \textbf{RQ.1 (Effectiveness)}.
  How effectively can errors be detected using incoherence alone without an oracle?
  What is the average error when incoherence is zero?
  How strong is the relationship between incoherence and error?
  \item \textbf{RQ.2 (Agreement)}.
  Does the result of an incoherence-based evaluation agree with the result of an error-based evaluation of LLMs?
  \item \textbf{RQ.3 (Ablation)}. How do incoherence and error vary as a function of a)~the number of synthesized programs, b)~the number of generated inputs, or c)~the temperature?
\end{itemize}

\subsection{Models and Datasets}
\label{section:experimental-setup:models-datasets}
\begin{table}[htbp]
\centering%
{\scriptsize
\setlength{\tabcolsep}{6pt}
\renewcommand{\arraystretch}{0.85}
\begin{tabular}{|l|l|}
\hline
Claude 4 Opus (2025/05/14)         & Claude 4 Sonnet (2025/05/14) \\
DeepSeek-Coder R1                  & DeepSeek-V3 (0324) \\
Gemini 2.0 Flash Lite              & Gemini 2.5 Pro (preview 05/06) \\
Gemini 2.5 Flash (preview 05/20)   & GPT-3.5 Turbo \\
GPT-4                              & GPT-4 Turbo \\
GPT-4o                             & GPT-o4 Mini \\
LLaMA 3.1 8B Instruct              & LLaMA 3.3 70B Instruct \\
LLaMA 4 Maverick 17B               & Ministral 8B \\
\hline
\end{tabular}
}
\caption{Large Language Models used in our experiments.}
\label{table:models}
\end{table}

\textbf{Models}. 
Table~\ref{table:models} shows the large language models (LLMs) used in our experiments. At the time of writing, these 16 LLMs represent the most successful LLMs for code generation according to several popular leaderboards \cite{leaderboards}. They also represent the current portfolio of the most popular LLM vendors: Anthropic, DeepSeek, Google, Meta, Mistral, and OpenAI.
By default, we chose a \emph{temperature of 0.6}, a value commonly used in prior work on code generation \cite{li2024honestcoder, deepseek}. We vary the temperature parameter in the ablation study (RQ3).

\textbf{Datasets}. We evaluate our measures of incoherence and error using the 16 LLMs on two (2) popular code generation benchmarks: {HumanEval} \cite{humaneval} and {MBPP} (Mostly Basic Python Problems) \cite{MBPP}.
\emph{HumanEval} is a human-written benchmark published by OpenAI in 2021, consisting of 164 programming tasks.
\emph{MBPP} is a crowd-sourced benchmark published by Google in 2022. We used the author-sanitized version of MBPP containing 426 hand-verified programming tasks.
For every task, they offer 
\begin{itemize}
\item a natural language description of the task $d$,
\item a ground-truth Python implementation $\gtfun{d}$, and
\item an average of 7.7 (and 3) Python test inputs for HumanEval (and MBPP, respectively).
\end{itemize}

\subsection{Variables and Measures}
\label{sec:variables}
Given a code generator $\coder$ and input generator $\gen$, programming task $d$, a query budget $m$ and a testing budget $n$, the \textbf{empirical error} $\hat \Err(d,m,n)$ on $d$ as estimator of the pointwise error is computed as 
\begin{equation}
    \hat{\Err}(d,m,n)=
    \frac{1}{n}
    \sum_{i=1}^n
    \mathbb{I}\bigl(\sem{\pi^d_{y_i}}(x^d_i) \neq \gtfun{d}(x^d_i)\bigr)
\end{equation}
and the \textbf{empirical incoherence} $\hat \Inc(d, m, n)$ on $d$ as estimator of the pointwise incoherence is computed as 
\begin{equation}
    \hat{\Inc}(d,m,n)=
    \frac{1}{n}
    \sum_{i=1}^n
    \mathbb{I}\bigl(\sem{\pi^d_{y_i}}(x_i^d) \neq \sem{\pi^d_{y'_i}}(x_i^d)\bigr)
\end{equation}
where $\pi^d_1,...,\pi^d_m$ are sampled from $\coder(d)$, $x_1^d,...,x_n^d$ are sampled from $\gen(d)$ and $y_1,...,y_n,y'_1, ...,y'_n$ are sampled from $\text{Uniform}(\{1,...,m\})$.

Given the set of programming tasks $S$, we can now write the \textbf{empirical \passone} score in terms of the empirical error (cf. Eq.~(\ref{eqn:one})):
\begin{align}
1-\frac{1}{|S|}\sum_{d\in S}\mathbb{I}\bigl(0\neq\hat\Err(d,1,n)\bigr)
\end{align}

The \textbf{mean empirical error} $\bar\Err(S,m,n)$ is
 $\sum_{d\in S}\frac{\hat\Err(d,m,n)}{|S|}$ while the \textbf{mean empirical incoherence} $\bar\Inc(S,m,n)$ is computed as
 $\sum_{d\in S}\frac{\hat\Inc(d,m,n)}{|S|}$.

The \textbf{detection rate} is the proportion of tasks with non-
zero emp. error that have a non-zero empirical incoherence, i.e.,
$
\frac{1}{|S_x|}\sum_{d\in S_x}\mathbb{I}\bigl(0\neq \hat\Inc(d,m,n)\bigr)
$
where $S_x=\{d\ |\ d\in S\wedge \hat\Err(d,m,n)\neq 0\}$.

The \textbf{undetected mean empirical error} is the mean empirical error of tasks with zero empirical incoherence, i.e.,
$\frac{1}{|S_u|}\sum_{d\in S_u}\hat\Err(d,m,n)$
where $S_u=\{d\ |\ d\in S\wedge \hat\Inc(d,m,n)= 0\}$.

We measure the \textbf{strength of the relationship} between two random variables, i.e., empirical incoherence and error, using Spearman's rank correlation coefficient $\rho$.
We measure the \textbf{agreement on ranking} when sorting the performance of LLMs measured by the proportion of programming tasks (a)~with zero mean empirical error versus (b)~with zero mean empirical incoherence, also using Spearman's rank correlation coefficient.

\subsection{Implementation}
\label{section:experimental-setup:implementation}

For each task, we generate $m$ candidate functions per coder (default $m=10$), using vendor APIs. Inputs are generated via mutation-based fuzzing ($n=1000$ by default). We estimate both error (w.r.t.\ ground truth $f^*_d$) and incoherence (between candidates), with all executions sandboxed (60s timeout). Experiments ran on AMD EPYC 7713P CPU (128 threads), 251~GB RAM. More details in Appendix \ref{sec:impl}.

\vspace{-0.2cm}
\section{Empirical Results}
\label{section:empirical-results}

\subsection*{RQ-1. Effectiveness}
\label{section:empirical-results:RQ1}
\begin{table}\scriptsize\centering
\begin{tabular}{@{}c@{ }c@{ }|l@{\quad}c@{ \ }|@{ \ }c@{\quad}c@{\quad}c@{\quad}c@{ }|@{}}
\cline{3-8}
&&&\textbf{Mean} &\textbf{Mean} &\textbf{Spearman}& \textbf{Detection} & \textbf{Undetected}\\
&&\textbf{LLM}& \textbf{Error} & \textbf{Incoherence} & \textbf{Correlation} & \textbf{Rate}	& \textbf{Mean Error}\\\hline
\multirow{3}{*}{ }&\multirow{3}{*}{\rotatebox[origin=c]{90}{\textbf{MBPP}}}
&  Code & 0.2960 & 0.0995 & 0.5276 & 0.6866 & 0.2071\\
& & General & 0.2773 & 0.1123 & 0.6105 & 0.7243 & 0.1638\\
& & Small & 0.3741 & 0.1641 & 0.5892 & 0.7037 & 0.2107\\\cline{3-8}
& & \textbf{Mean} & 0.3009 & 0.1203 & 0.5621 & 0.6857 & 0.1866\\
\hline\hline
\multirow{3}{*}{\rotatebox[origin=c]{90}{\textbf{Human}}}&\multirow{3}{*}{\rotatebox[origin=c]{90}{\textbf{Eval}}}
& Code & 0.0763 & 0.0295 & 0.7171 & 0.7188 & 0.0417\\
& &  General & 0.0927 & 0.0483 & 0.7181 & 0.7042 & 0.0460\\
& & Small & 0.1585 & 0.0911 & 0.7282 & 0.7381 & 0.0737\\\cline{3-8}
& & \textbf{Mean} & 0.1050 & 0.0560 & 0.6861 & 0.6616 & 0.0471\\
\hline
\end{tabular}
\caption{Performance of 3 LLMs on 2 benchmarks.
The \textbf{mean} is reported across all 16 LLMs; \texttt{Code} = Gemini 2.5 Pro, \texttt{General} = Gpt--4o, and \texttt{Small} = Ministral 8b.}
\label{tab:rq1}
\end{table}

Table~\ref{tab:rq1} shows the results for across all 16 LLMs (mean) and for three representative models (code, general, and small) for both code generation benchmarks. Table~4 (appendix) shows the results for all 16 LLMs. The measures in the header row are discussed in Section~\ref{sec:variables}. Figure~\ref{fig:rq1-relationship} shows a scatter plot illustrating the relationship between error and incoherence.

\textbf{Results}. A non-zero incoherence \emph{effectively} detects a non-zero error without access to a ground truth implementation. The mean detection rate across all 16 LLMs for MBPP and HumanEval are 69\% and 66\%, respectively. There does not seem to be a substantial difference in detection rate between the code-generation specific LLM (Gemini 2.5 Pro) and the general-purpose or the small LLM (GPT-4o, Mistral 8b).
In cases where the incoherence is zero, the mean error is also substantially lower. Concretely, the mean error reduces from 30\% to 19\% for MBPP and from 11\% to 5\% for HumanEval (``Undetected Mean Error"). For the small LLM, the mean error is generally higher than for the other LLMs, but the percentage decrease when incoherence is zero is similar.

Figure~\ref{fig:rq1-relationship} best illustrates the relationship between error and incoherence. We can clearly see the consequence of the inequality in Theorem~\ref{theorem:incoherence-inequality}. Error is usually greater than incoherence for a programming task. There are some tasks (on the left of each plot) where incoherence is zero but the error is non-zero.
For the average LLM, we find a \emph{moderate correlation} between error and incoherence (0.56 for MBPP; 0.69 for HumanEval). In Table~\ref{tab:rq1}, for the three LLMs evaluated on HumanEval, we even find a \emph{strong correlation} (gt. 0.7).

\result{\textbf{RQ-1}. \emph{A non-zero incoherence effectively detects about two-thirds of the non-zero errors in the absence of a ground truth implementation. In cases where the incoherence is zero, the mean error is substantially lower than the average. The plot of error and incoherence provides empirical confirmation for our inequality.}}\vspace{-0.1cm}

\begin{figure}\centering
\includegraphics[width=0.49\columnwidth, clip]{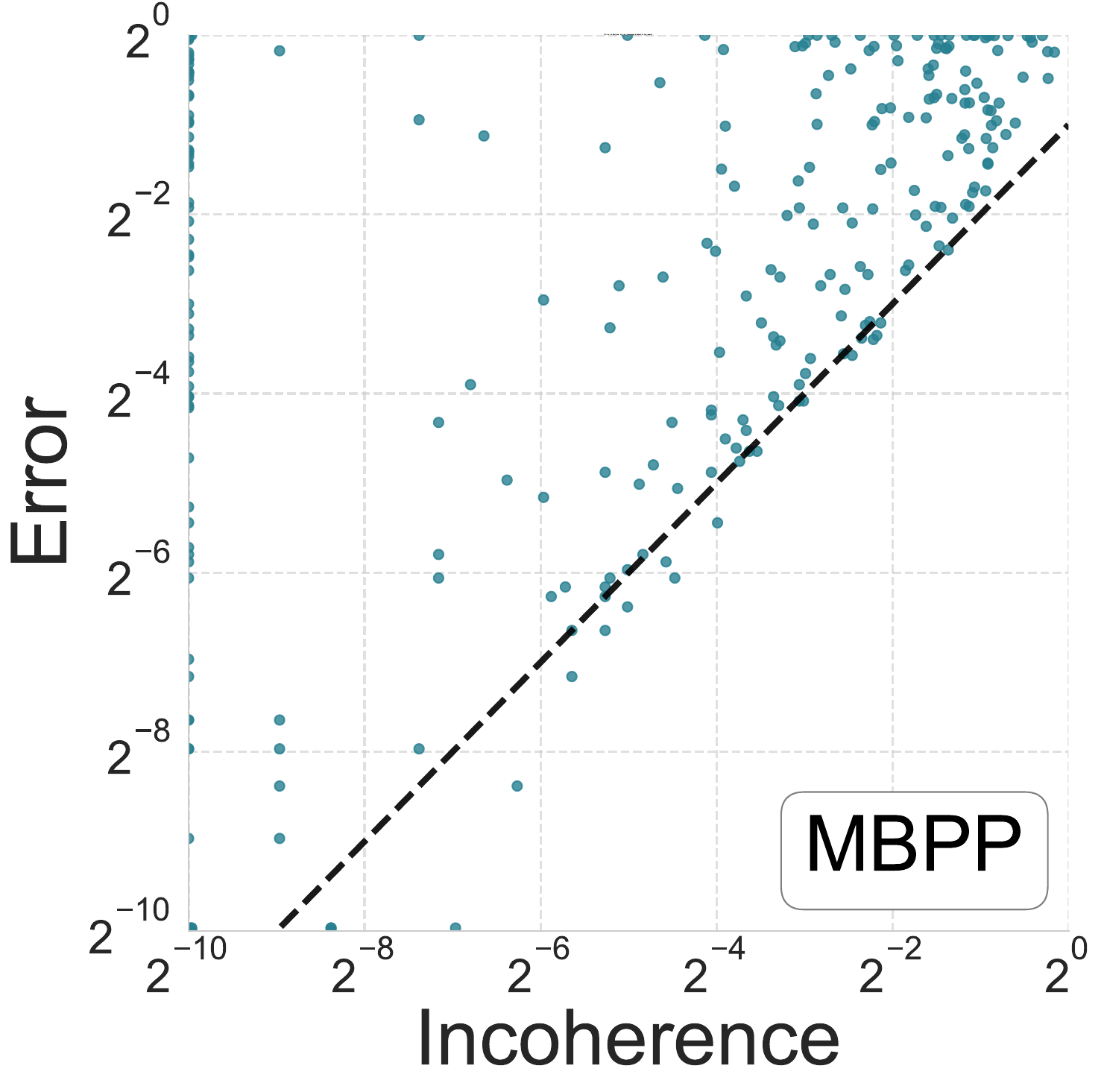}
\includegraphics[width=0.49\columnwidth, clip]{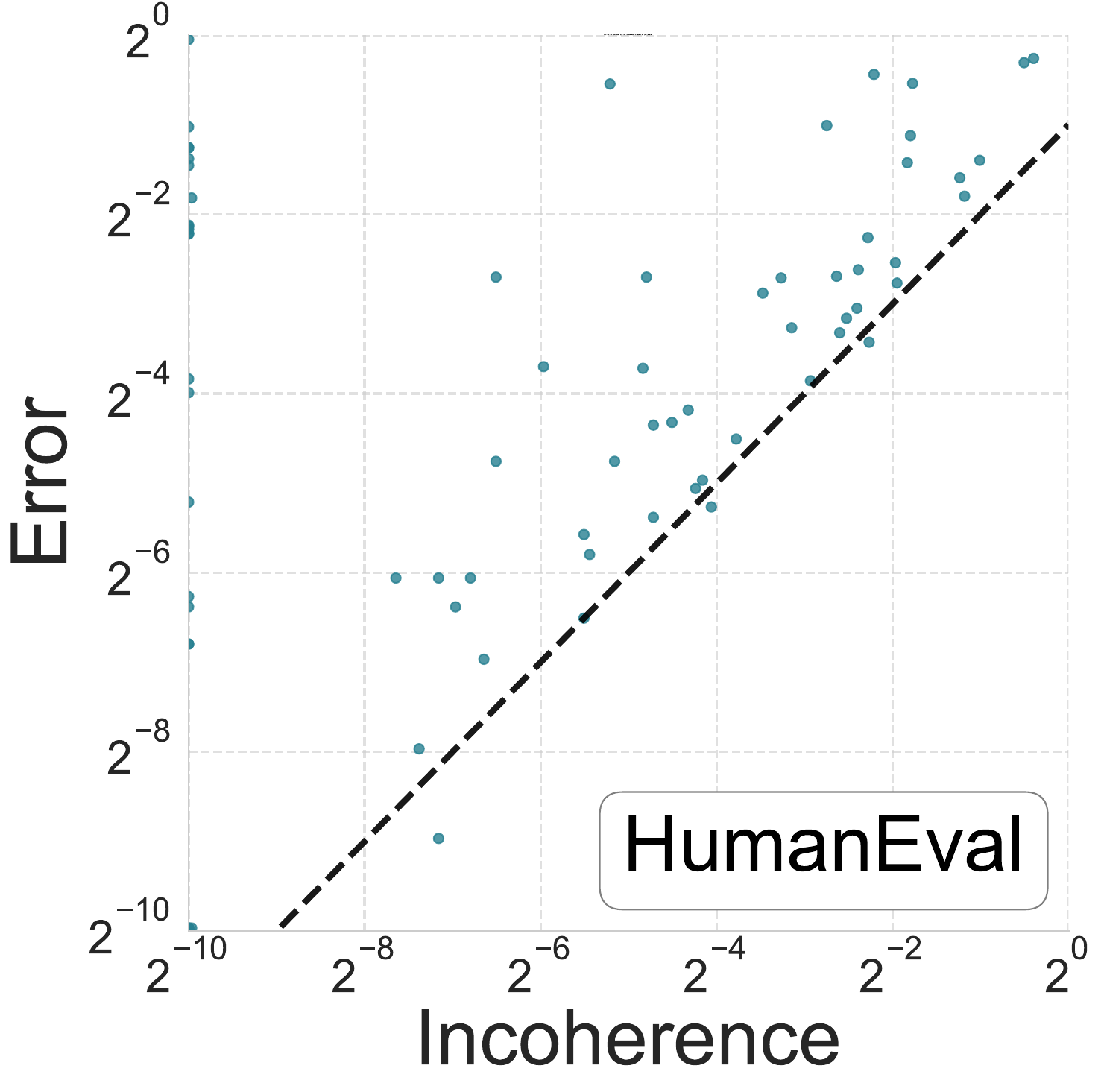}
\caption{Relationship between error and incoherence for GPT-4o on MBPP and HumanEval benchmarks. The dashed line demonstrates the inequality in Theorem~\ref{theorem:incoherence-inequality}.}
\label{fig:rq1-relationship}
\end{figure}
\subsection*{RQ-2. Agreement on LLM Ranking}
\label{section:empirical-results:RQ2}

If the incoherence- and error-based rankings of LLMs agree, we can reliably substitute one measure for the other. We could remove the requirement to provide painstakingly manually-written ground-truth implementations for every programming task when constructing new code generation benchmarks. We could \emph{mitigate critical threats to validity} in benchmarking of new LLM-based code generation systems, such as overfitting or data leakage.

Figure~\ref{fig:mainresult} (on the title page) shows a scatter plot of the ranking of all 16 LLMs in terms of the number of projects with zero error (i.e., \passone; cf. Eqn.~(\ref{eq:passone})) versus the ranking of the same LLMs in terms of the number of projects with our non-zero oracle-less incoherence measure. 

\textbf{Results}. We observe a \emph{very strong agreement} on the rankings. Concretely, $\rho = 0.92$ and $\rho =0.94$ for MBPP and HumanEval, respectively, at a significance level $p<0.0001$. The rankings are close to the diagonal. 

\result{\textbf{RQ-2}. \emph{An oracle-based evaluation can be reliably substituted by an incoherence-based evaluation.
Specifically, there is a \emph{very strong agreement} between the rankings of LLMs in terms of the proportion of programming tasks that are considered correct (a)~via the lack of a pointwise difference with a ground truth implementation (as in, \passone) versus (b)~via the lack of a pointwise difference between two randomly generated solutions.}}

\begin{table}\footnotesize\centering%
\begin{tabular}{@{}r@{ }l|rc|rc@{}}
&& \multicolumn{2}{c|}{\emph{MBPP}} & \multicolumn{2}{c}{\emph{HumanEval}}\\\hline
& &\multicolumn{1}{c}{\textbf{Expenses}}& \textbf{Detection} & \multicolumn{1}{c}{\textbf{Expenses}} & \textbf{Detection}\\
&  & \textbf{(in USD)} & \textbf{Rate} & \textbf{(in USD)} & \textbf{Rate}\\\hline
$m=$&$1$ & 0.8730 & 0.0000 & 0.4436 &  0.0000 \\
$m=$&$2$ & 1.7557 & 0.3974 & 0.8904 &  0.3333 \\
$m=$&$5$ & 4.3992 & 0.6357 & 2.2189 &  0.5846 \\
$m=$&$10$ & 8.8138 & 0.7243 & 4.4530 &  0.7042 \\
$m=$&$25$ & 22.0332 & 0.7742 & 11.1055 & 0.8267 \\
$m=$&$50$ & 43.9539 & 0.8105 & 22.2106 & 0.8182 \\\hline
\multicolumn{6}{@{}p{0.96\columnwidth}@{}}{\vspace{-0.15cm}\emph{\scriptsize (a) Detection rate and LLM costs as the query budget, i.e., the number $m$ of programs generated by $\coder(d)$ increases.}\vspace{0.1cm}}\\\hline
$t=$&$0.2$ & 8.8138 & 0.5422 & 4.4530 & 0.5556 \\
$t=$&$0.6$ & 8.8174 & 0.7148 & 4.4840 & 0.7286 \\
$t=$&$1$ & 9.0657 & 0.8050 & 4.5254 & 0.7600 \\\hline
\multicolumn{6}{@{}p{0.96\columnwidth}@{}}{\vspace{-0.15cm}\emph{\scriptsize (b) Detection rate and LLM costs as temperature $t$ of $\coder(d)$ increases.}\vspace{0.1cm}}\\\hline
$n=$&$100$ & 8.8174  & 0.6811 & 4.4840  & 0.6615 \\
$n=$&$1000$ & 8.8174  & 0.7148 & 4.4840  & 0.7286 \\
$n=$&$2000$ & 8.8174  & 0.7200 & 4.4840  & 0.7083 \\
$n=$&$5000$ & 8.8174  & 0.7355 & 4.4840  & 0.7222 \\
$n=$&$10000$ & 8.8174  & 0.7445 & 4.4840  & 0.7222 \\\hline
\multicolumn{6}{@{}p{0.96\columnwidth}@{}}{\vspace{-0.15cm}\emph{\scriptsize (c) Detection rate and LLM costs as the testing budget, i.e., the number $n$ of test inputs generated by $\gen(d)$ increases.}}
\end{tabular}
\caption{Results of our ablation study. We vary one value while keeping all others constant. $\coder$ is GPT-4o. Default number of programs per task: $m=10$. Default number of test inputs per task: $n=1000$. Default temperature: $t=0.6$.}
\label{tab:rq3}
\end{table} 

\vspace{-0.1cm}
\subsection*{RQ-3. Ablation}
\label{section:empirical-results:RQ3}

Table~\ref{tab:rq3} shows the results of our ablation study as we vary the LLM the number $m$ of generated programs, the number $n$ of generated test inputs, or the LLM's temperature $t$. A higher temperature increases the likelihood that the LLM samples a lower-probability token during next-token prediction.
We vary one parameter and keep all others constant ($m=10$, $n=1000$, $t=0.6$, GPT4-o; cf. \S\ref{section:experimental-setup}).

\textbf{Query budget} $m$. The detection rate increases with $m$. For instance, when changing $m$ from 10 to 50, the detection rate increases by 14--19\% from 0.7 to about 0.83 for HumanEval and from 0.72 to 0.82 for MBPP. The advantage comes at a substantial monetary cost. When changing $m$ from 10 to 50, our expenses increased by more than fivefold, e.g., from \$9 to \$44 for MBPP.
For HumanEval, we actually observe a slightly higher detection rate (0.83) at $m=25$, which we explain by the randomness of the sampling and test generation process. It might also indicate that the detection rate starts to saturate for larger values of $m$ (which we determined as uneconomical for us to test).
Another interesting observation is that just sampling a second program ($m=2$) already gives us a 0.33 to 0.4 detection rate.

\textbf{Temperature} $t$. Detection rate increases with $t$. For instance, when changing $t$ from 0.2 to 1.0, we see detection rate increase by 36--50\% from 0.54 to 0.81 for MBPP and from 0.56 to 0.76 for HumanEval. A high temperature induces a high output diversity, which seems to increase the LLM's incoherence, which serves us well in error detection.

\textbf{Test inputs} $n$. Detection rate increases with $n$. However, compared to the other hyperparameters, a substantial increase in the number of generated test inputs induces only a relatively small increase in detection rate.

\vspace{-0.2cm}
\result{\textbf{RQ-3}. \emph{Increasing $\coder$'s query budget $m$, $\gen$'s testing budget $n$, or the temperature $t$ also increases the detection rate. However, an x-fold increase in query budget comes at a greater-than x-fold increase in expenses.}}

\section{Threats to Validity}
\label{section:threats-to-validity}
As with any empirical study, there are threats to the validity of our results and conclusions. The first threat is to the \emph{external validity}, i.e., the extent to which our findings can be generalized. As the subjects of our study, we selected LLMs from all major LLM vendors that were top-performing according to code generation leaderboards. They represent the current state-of-the-art.
As the objects of our study, we selected the two most widely used code generation benchmarks, MBPP and HumanEval, to facilitate comparison with results in related research. However, the findings may not generalize to more complex programming tasks or programming languages other than Python, and we call on the community to replicate our experiments for their use cases. Beyond the empirical results we also formally prove certain properties of incoherence and its estimation in the general.

The second threat is to the internal validity, i.e., the extent to which the presented evidence supports our claims about cause and effect within the context of our study. In the benchmarks, the task description may be ambiguous or the ground-truth implementation incorrect \cite{qualityassessment}. We use popular well-scrutinized benchmarks. From MBPP, we chose the best-quality, hand-curated set of tasks.
\toolname may contain bugs itself, but we release all our scripts and data for the community to scrutinize.

\section{Perspective}
We believe that our incoherence-based perspective gives rise to a proliferation of new techniques built for \emph{trustworthy code generation with probabilistic guarantees}.

In this paper, we discuss the formal estimation of the correctness of an LLM-generated program when there is \emph{no automated mechanism} to decide whether a program is correct or not (e.g., a formal specification or a ground-truth implementation). We model the generated program as a random variable drawn from an unknown distribution induced by the coder (e.g., the LLM). This opens the door for a probabilistic notion of correctness. Our measure, incoherence, formalizes the observation that, if two random programs for the same task disagree on the output for an input, at least one must be incorrect. We formally demonstrate how the coder's error on a task has a lower bound in terms of the coder's incoherence on that task and empirically observe that a non-zero incoherence detects more than two-thirds of the incorrect programs. 
Since incoherence does not depend 
on model internals, it can be applied broadly across LLMs of any kind and even to probabilistic systems like LLM-agents.

An incoherence-based evaluation of the code generation capabilities of multiple LLMs also addresses several open challenges of the traditional ground-truth-based \passone evaluation. The existing process of curating large coding benchmarks with correct human-generated ground-truth implementations is labour-intensive, error-prone, and subject to future data leakage issues \cite{memorizing}. Incoherence paves the way for evaluations on a substantially larger scale, basically on a stream of programming tasks.

\section*{Acknowledgments}
We thank the anonymous reviewers for their constructive feedback
and for helping us improve this paper. This research is partially
funded by the European Union. Views and opinions expressed are
however those of the author(s) only and do not necessarily reflect
those of the European Union or the European Research Council
Executive Agency. Neither the European Union nor the granting
authority can be held responsible for them. This work is supported
by ERC grant (Project AT\_SCALE, 101179366). This research is also
partially funded by the ENS Paris-Saclay who graciously supported the internship
of the first author at the MPI for Security and Privacy.


\bibliography{aaai25}


\appendix

\section{Implementation Details}\label{sec:impl}
Figure~\ref{fig:diagram} provides a procedural overview of our Python implementation, called \toolname. For every programming \emph{task} in a dataset, for every \emph{coder} (i.e., LLM), repeated $M$ times to produce $M$ \emph{candidate functions}, \toolname uses the LLM vendor-provided application programming interface (API) to generate a Python program that implements the natural language specification $d$ that is provided with the task. The coder's prompt further contains instructions to adhere to a given function signature. To optimize throughput, \toolname dispatches LLM queries in parallel whenever possible, with a fallback to sequential execution in the event of API rate limiting or errors. By default, we generate $m=10$ candidate functions for each task. We vary $m\in \{1,2,5,10,25,50\}$ in the ablation study (RQ3).

For every programming task, once for the empirical error and once for the empirical incoherence, \toolname uses the \emph{input generator} to generate $n$ inputs for the generated \emph{candidate functions} using the task-provided \emph{seed inputs}. The test generator implements the budget-constrained $\coder_m$ presented in Section~\ref{sec:fixedbudget} and the type-aware mutation-based fuzzing method $\gen$ introduced in EvalPlus~\cite{EvalPlus} and discussed in Section~\ref{sec:inputgeneration}. We provide a list of supported mutations in the appendix (Table~\ref{tab:mutation_types}). We note that non-zero incoherence implies non-zero error for all $\gen$ (Thm.~\ref{theorem:incoherence-inequality} and Eqn.~(\ref{eqn:seven})). We also note that point-wise error [Eqn.~(\ref{eqn:pntwiseError})] and incoherence [Eqn.~(\ref{eqn:seven})] are defined w.r.t. the same distribution induced by $\gen$.
To ensure robustness, all executions, whether for compilation, incoherence estimation, or error estimation, are subject to a 60-second timeout. 
The empirical error is computed using the task-provided ground-truth (GT) function $f^*_d$. By default, we generate $n=1000$ test inputs.  We vary $n\in \{100,1000,2000,5000,10000\}$ in RQ3.

\begin{figure*}[t]
    \centering
    \includegraphics[width=\textwidth]{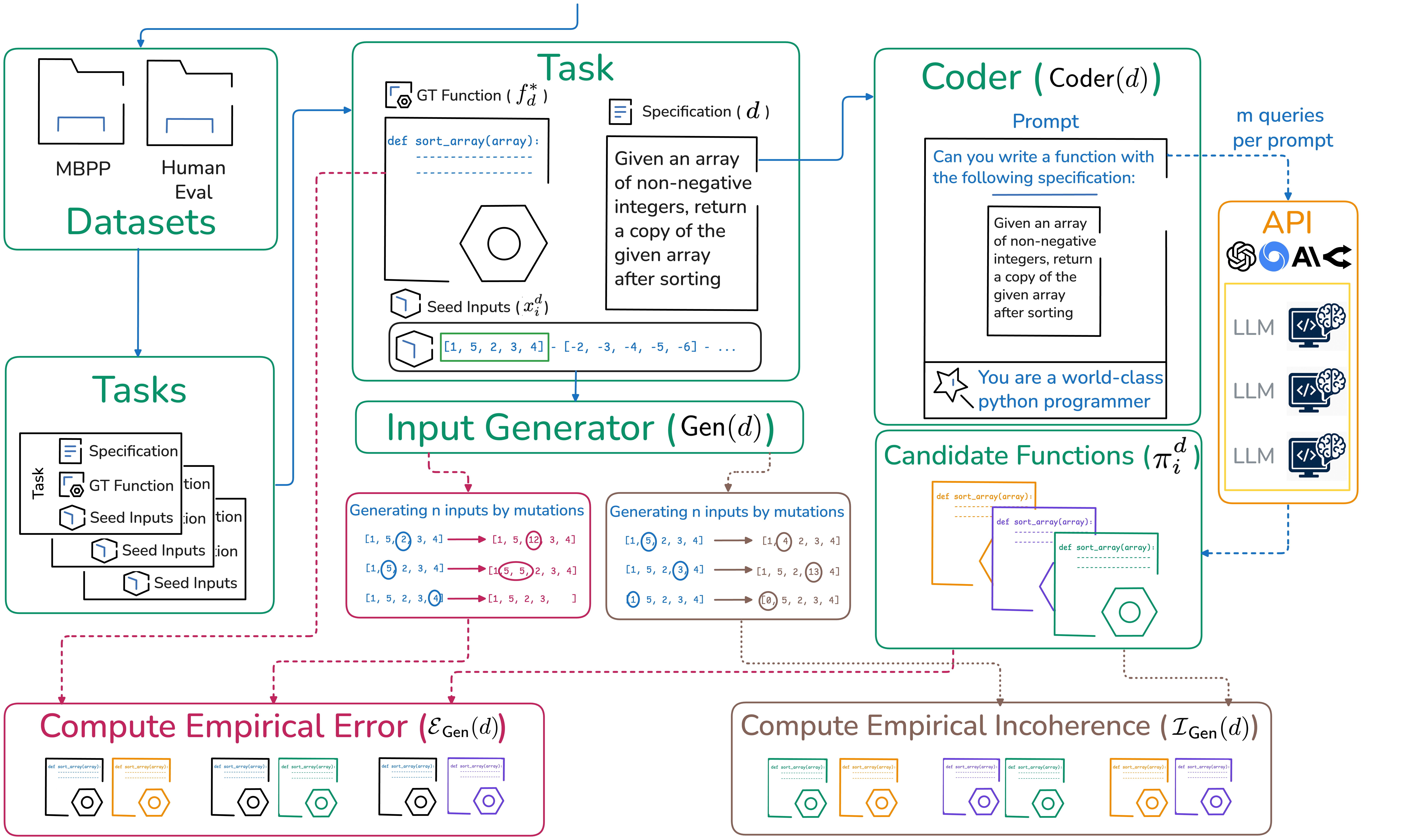}
    \caption{Workflow of our implementation \toolname.}
    \label{fig:diagram}
\end{figure*}

We publish all data, the analysis, and the virtual experimental infrastructure to reproduce our experiments:\\\artifact 
\subsection{Infrastructure}
\label{section:experimental-setup:infrastructure}
We used a single machine equipped with an AMD EPYC 7713P 64-Core Processor (128 threads), 251 GB of RAM, running Ubuntu 22.04.5 LTS 64-bit.

\section{Proof of Pointwise Incoherence Inequality}
\label{appendix:proof-pointwise-incoherence-ineq}

For readers who prefer a fully explicit probability-theoretic formalization, we describe here the underlying probability space used implicitly in the main text.

We consider a probability space $(\Omega, \mathcal{F}, \Prob)$, where $\Omega$ is the sample space, $\mathcal{F}$ the $\sigma$-algebra of events, and $\Prob$ the probability measure. Random variables are measurable functions from $\Omega$ to their codomain. For a mathematical expression $expr[X_1,\dots,X_n]$ involving random variables, we write
\[
\{ expr \} := \{ \omega \in \Omega \mid expr[X_1(\omega),\dots,X_n(\omega)]\text{ holds} \}.
\]

Theorem~\ref{theorem:incoherence-inequality} states that :
\[
\forall \gen, \forall d \in \Descr,\quad \Inc_\gen(d) \leq 2 \times \Err_\gen(d).
\]

\begin{proof}
Let $\gen$ be an input generation system and let $d\in \Descr$.
Let $\ProgRV^d_1, \ProgRV^d_2 \overset{iid}{\sim} \coder(d)$ represent two independently sampled programs from $\coder$ output distribution for task $d$ and let $X \sim \gen(d)$ represent an input sampled from $\gen$ for task $d$.

For every sample $\omega \in \SampleSpace$, let $f_1 = \sem{\ProgRV^d_1(\omega)}$, $f_2 = \sem{\ProgRV^d_2(\omega)}$ and $x = X(\omega)$ be the corresponding outcomes of $\sem{\ProgRV^d_1}, \sem{\ProgRV^d_2}$ and $X$. Then we have the following implication
\begin{equation*}
    f_1(x) \neq f_2(x) \implies (f_1(x) \neq \gtfun{d}(x) \lor f_2(x) \neq \gtfun{d}(x)).
\end{equation*}
Taking corresponding probabilistic events
\begin{align*}
    &E_{\ProgRV^d_1,\ProgRV^d_2} := \{\sem{\ProgRV^d_1}(X)\neq\sem{\ProgRV^d_2}(X)\} \\
    &E_{\ProgRV^d_1,\gtfun{d}} := \{\sem{\ProgRV^d_1}(X)\neq \gtfun{d}(X)\} \\
    &E_{\ProgRV^d_2,\gtfun{d}} := \{\sem{\ProgRV^d_2}(X)\neq \gtfun{d}(X)\}
\end{align*}
We thus have
\begin{equation*}
    E_{\ProgRV^d_1,\ProgRV^d_2} \subseteq E_{\ProgRV^d_1,\gtfun{d}} \cup E_{\ProgRV^d_2,\gtfun{d}}.
\end{equation*}
Therefore
\begin{align*}
    \Prob(E_{\ProgRV^d_1,\ProgRV^d_2}) & \leq \Prob(E_{\ProgRV^d_1,\gtfun{d}} \cup E_{\ProgRV^d_2,\gtfun{d}}) \\
    & \leq \Prob(E_{\ProgRV^d_1,\gtfun{d}}) + \Prob(E_{\ProgRV^d_2,\gtfun{d}}).
\end{align*}
Since by definition, $\Inc_\gen(d) = \Prob(E_{\ProgRV^d_1,\ProgRV^d_2})$ and $\Err_\gen(d) = \Prob(E_{\ProgRV^d_1,\gtfun{d}}) = \Prob(E_{\ProgRV^d_2,\gtfun{d}})$.
We thus obtain
\[
\Inc_\gen(d) \leq 2 \times \Err_\gen(d).
\]

\end{proof}

\section{Estimation of Incoherence and Error}
\label{appendix:estimation-inc-err}

\subsection{Monte Carlo Estimation of Pointwise Error}
\label{appendix:estimation-inc-err:monte-carlo-estimation-err}

Given $\coder$ a code generation system and provided $\gen$ an input generation system, when the ground truth implementation $\gtfun{d}$ is available, we can estimate the pointwise error $\Err_\gen(d)$ using a Monte Carlo procedure based on Definition~\eqref{eqn:pntwiseError} as illustrated in Algorithm~\ref{alg:point-err-estimation}.

\begin{algorithm}[tb]
\caption{Monte Carlo Estimation of $\Err_\gen(d)$}
\label{alg:point-err-estimation}
\textbf{Input}: $\coder$, $\gen$, task $d \in \Descr$, ground truth $\gtfun{d}$, error tolerance $\epsilon > 0$, failure probability $\delta > 0$\\
\textbf{Output}: $\bar{\Err}_\gen(d)$ estimate such that $|\bar{\Err}_\gen(d) - \Err_\gen(d)| \leq \epsilon$ with probability $\geq 1 - \delta$
\begin{algorithmic}[1]
\STATE Let $N = \left\lceil \frac{\log(2/\delta)}{2\epsilon^2} \right\rceil$
\FOR{$i = 1$ to $N$}
    \STATE Sample $\prog \sim \coder(d)$
    \STATE Sample $x \sim \gen(d)$
    \STATE $e_i \gets \mathbb{I}\bigl(\sem{\prog}(x) \neq \gtfun{d}(x)\bigr)$
\ENDFOR
\STATE \textbf{return} $\bar{\Err}_\gen(d) = \frac{1}{N} \sum_{i=1}^N e_i$
\end{algorithmic}
\end{algorithm}

\paragraph{Correctness Guarantee.}
Each $e_i$ is a Bernoulli random variable with $\mathbb{E}[e_i] = \Err_\gen(d)$. By Hoeffding’s inequality:
\[
\Prob(|\bar{\Err}_\gen(d) - \Err_\gen(d)| \geq \epsilon) \leq 2\exp(-2N\epsilon^2).
\]
Thus, with $N \ge \frac{\log(2/\delta)}{2\epsilon^2}$, we obtain the desired PAC guarantee.

\subsection{Monte Carlo Estimation of Pointwise Incoherence}
\label{appendix:estimation-inc-err:monte-carlo-estimation-inc}

Given $\coder$ a code generation system and provided $\gen$ an input generation system, even if the ground truth implementation $\gtfun{d}$ is unavailable, we can estimate the pointwise incoherence $\Inc_\gen(d)$ using a Monte Carlo procedure based on Definition~\eqref{eqn:pntwiseIncoherence} as illustrated in Algorithm~\ref{alg:point-inc-estimation}.

\begin{algorithm}[tb]
\caption{Monte Carlo Estimation of $\Inc_\gen(d)$}
\label{alg:point-inc-estimation}
\textbf{Input}: $\coder$, $\gen$, task $d \in \Descr$, error tolerance $\epsilon > 0$, failure probability $\delta > 0$\\
\textbf{Output}: $\bar{\Inc}_\gen(d)$ estimate such that $|\bar{\Inc}_\gen(d) - \Inc_\gen(d)| \leq \epsilon$ with probability $\geq 1 - \delta$
\begin{algorithmic}[1]
\STATE Let $N = \left\lceil \frac{\log(2/\delta)}{2\epsilon^2} \right\rceil$
\FOR{$i = 1$ to $N$}
    \STATE Sample $\prog_1, \prog_2 \overset{iid}{\sim} \coder(d)$
    \STATE Sample $x \sim \gen(d)$
    \STATE $d_i \gets \mathbb{I}\bigl(\sem{\prog_1}(x) \neq \sem{\prog_2}(x)\bigr)$
\ENDFOR
\STATE \textbf{return} $\bar{\Inc}_\gen(d) = \frac{1}{N} \sum_{i=1}^N d_i$
\end{algorithmic}
\end{algorithm}

\paragraph{Correctness Guarantee.}
Each $d_i$ is a Bernoulli random variable with $\mathbb{E}[d_i] = \Inc_\gen(d)$. Applying Hoeffding’s inequality:
\[
\Prob(|\bar{\Inc}_\gen(d) - \Inc_\gen(d)| \geq \epsilon) \leq 2\exp(-2N\epsilon^2).
\]
Hence, with $N \ge \frac{\log(2/\delta)}{2\epsilon^2}$, we obtain the desired PAC guarantee.

\subsection{PAC Detection of Nonzero Incoherence}
\label{appendix:estimation-inc-err:pac-detection}

Given $\coder$ a code generation system and provided $\gen$ an input generation system, we may want to detect with high confidence whether the model exhibits nonzero pointwise incoherence on a given task $d \in \Descr$. We present a probabilistically sound decision procedure based on repeated disagreement tests, as illustrated in Algorithm~\ref{alg:pac-incoherence-detection}.

\begin{algorithm}[tb]
\caption{PAC Detection of non-zero $\Inc_\gen(d)$}
\label{alg:pac-incoherence-detection}
\textbf{Input}: $\coder$, $\gen$, task $d \in \Descr$, incoherence threshold $\epsilon > 0$, confidence parameter $\delta > 0$\\
\textbf{Output}: \texttt{true} if $\Inc_\gen(d) > 0$ is detected, otherwise \texttt{false}
\begin{algorithmic}[1]
\STATE Let $N = \left\lceil \frac{\log(\delta)}{\log(1 - \epsilon)} \right\rceil$
\FOR{$i = 1$ to $N$}
    \STATE Sample $\prog_1, \prog_2 \overset{iid}{\sim} \coder(d)$
    \STATE Sample $x \sim \gen(d)$
    \IF{$\sem{\prog_1}(x) \neq \sem{\prog_2}(x)$}
        \STATE \textbf{return} \texttt{true}
    \ENDIF
\ENDFOR
\STATE \textbf{return} \texttt{false}
\end{algorithmic}
\end{algorithm}

\paragraph{Correctness Guarantee.}
Let $\epsilon > 0$ and $\delta > 0$. Then Algorithm~\ref{alg:pac-incoherence-detection}, when run with parameters $\epsilon$ and $\delta$, satisfies the following:

\begin{itemize}
    \item If $\Inc_\gen(d) = 0$, the algorithm always returns \texttt{false}.
    \item If $\Inc_\gen(d) \ge \epsilon$, the algorithm returns \texttt{true} with probability at least $1 - \delta$.
\end{itemize}

In other words, the algorithm detects non-zero incoherence with confidence at least $1 - \delta$, and it never returns false positives.

\begin{proof}
Let $\epsilon > 0$ and $\delta > 0$, and let $N = \left\lceil \frac{\log(\delta)}{\log(1 - \epsilon)} \right\rceil$. Then Algorithm~\ref{alg:pac-incoherence-detection} satisfies the following:

\begin{itemize}
    \item \textbf{No false positives.} If $\Inc_\gen(d) = 0$, then all sampled programs agree on all inputs almost surely. Therefore, no disagreement can ever be observed, and the algorithm always returns \texttt{false}.
    
    \item \textbf{False negative probability $\le \delta$.} Suppose instead that $\Inc_\gen(d) \ge \epsilon$. Then in each trial of the algorithm, the probability of observing a disagreement is at least $\epsilon$. Since the trials are independent, the probability that all $N$ trials fail to detect a disagreement is at most:
    \[
    (1 - \epsilon)^N.
    \]
    By the choice of $N$, this is at most $\delta$:
    \[
    (1 - \epsilon)^N \le \delta.
    \]
    Hence, the probability that the algorithm detects a disagreement and returns \texttt{true} is at least $1 - \delta$.
\end{itemize}
\end{proof}

\paragraph{Implication for Error Detection.}
Although the algorithm only checks for non-zero \emph{incoherence}, we can derive a guarantee for non-zero \emph{error} via Theorem~\ref{theorem:incoherence-inequality}, which states:
\[
\Inc_\gen(d) \le 2 \cdot \Err_\gen(d).
\]
Thus, a detection of $\Inc_\gen(d) \ge \epsilon$ implies that $\Err_\gen(d) \ge \epsilon / 2$. Consequently, this detection method provides a statistically sound way to identify non-zero errors without requiring access to a ground truth implementation.

\section{Functional Incoherence}
\label{appendix:functional-incoherence}

While pointwise incoherence measures this divergence over specific inputs, functional incoherence extends the idea to the full input space. That is, two implementations are functionally incoherent if they disagree on \emph{any} input. This notion aligns with the classic definition of program equivalence.

Let $\ProgRV_1^d, \ProgRV_2^d \overset{iid}{\sim} \coder(d)$ be two independently sampled programs for a task $d \in \Descr$. Let $\sem{\ProgRV^d_1}, \sem{\ProgRV^d_2}$ denote their functional interpretations. We define the \textbf{functional incoherence} of a code generation system $\coder$ for a task $d$ is the probability that two independently sampled implementations are not functionally equivalent:
\begin{equation}
\Inc(d) := \Prob(\sem{\ProgRV_1^d} \neq \sem{\ProgRV_2^d})
\end{equation}

This notion captures global behavioral disagreement between sampled programs. In practice, while direct evaluation of functional equivalence is undecidable in general, functional incoherence may be conservatively approximated via testing or symbolic execution.

Functional incoherence can be seen as the limiting form of pointwise incoherence. Specifically, $\Inc_\gen(d)$ (as defined in the main text) estimates incoherence over a distribution of inputs, whereas $\Inc(d)$ considers disagreement over the entire input space.

Importantly, just as pointwise incoherence lower-bounds pointwise error, we can show that functional incoherence provides a lower bound on functional error.

\subsection{Functional Incoherence as a Lower Bound on Error}
\label{appendix:functional-incoherence:inequality}

\begin{theorem}[Functional Incoherence Inequality]
\label{theorem:functional-incoherence-inequality}
For any task $d \in \Descr$, functional incoherence provides a lower bound on functional error:
\[
\Inc(d) \leq 2 \times \Err(d).
\]
\end{theorem}

\begin{proof}
Let $\gtfun{d}$ denote the ground truth implementation for task $d$.

Let $f_1 := \sem{\ProgRV_1^d}$ and $f_2 := \sem{\ProgRV_2^d}$ be two independently sampled implementations from $\coder(d)$. Then:
\[
\{f_1 \neq f_2\} \subseteq \{f_1 \neq f^\ast_d\} \cup \{f_2 \neq f^\ast_d\}
\]
This holds because if two functions differ, at least one must differ from the ground truth. Taking probabilities:
\begin{align*}
\Prob(f_1 \neq f_2) & \leq \Prob(f_1 \neq f^\ast_d) + \Prob(f_2 \neq f^\ast_d)
\end{align*}
Thus:
\[
\Inc(d) \leq 2 \cdot \Err(d)
\]
\end{proof}

\paragraph{Discussion.}
This result mirrors the pointwise inequality established in the main text (Theorem~\ref{theorem:incoherence-inequality}) and shows that disagreement between implementations is a rigorous signal of potential failure—even in the absence of ground truth.

While $\Err(d)$ requires access to $\gtfun{d}$, $\Inc(d)$ is fully estimable from samples of $\coder$. This makes functional incoherence a valuable unsupervised proxy for reliability, particularly in settings where model outputs must be audited without labeled data.

\clearpage
\begin{table*}[ht]\footnotesize\centering
\caption{Performance of 16 LLMs on two benchmarks. The final row in each benchmark section reports the mean performance across all 16 models.}
\begin{tabular}{lc|cccc}
\toprule
&\textbf{Mean} &\textbf{Mean} &\textbf{Spearman}& \textbf{Detection} & \textbf{Undetected}\\
\textbf{Model}& \textbf{Error} & \textbf{Incoherence} & \textbf{Correlation} & \textbf{Rate}	& \textbf{Mean Error}\\
\midrule
\textbf{MBPP}                      &        &        &        &        &       \\
GPT-4o                             & 0.2773 & 0.1123 & 0.6105 & 0.7243 & 0.1638\\
Claude 4 Opus (2025/05/14)         & 0.2715 & 0.0583 & 0.4508 & 0.4815 & 0.2140\\
Gemini 2.5 Pro (preview 05/06)     & 0.2960 & 0.0995 & 0.5276 & 0.6866 & 0.2071\\
LLaMA 4 Maverick 17B               & 0.2980 & 0.0680 & 0.4291 & 0.4741 & 0.2352\\
Claude 4 Sonnet (2025/04/14)       & 0.2693 & 0.0452 & 0.4151 & 0.4016 & 0.2150\\
DeepSeek-V3 (Mar 2024)             & 0.2659 & 0.1185 & 0.6112 & 0.7406 & 0.1555\\
GPT-o4 Mini                        & 0.3109 & 0.1045 & 0.5646 & 0.7543 & 0.1784\\
Mistral 8B                         & 0.3741 & 0.1641 & 0.5892 & 0.7037 & 0.2107\\
LLaMA 3 70B Instruct               & 0.3574 & 0.1689 & 0.5419 & 0.8139 & 0.2288\\
DeepSeek-Coder R1                  & 0.2286 & 0.0811 & 0.5650 & 0.6453 & 0.1640\\
Gemini 2.5 Flash (preview 05/20)   & 0.2949 & 0.0936 & 0.5216 & 0.7122 & 0.2183\\
Gemini 2.0 Flash Lite              & 0.2913 & 0.1097 & 0.6237 & 0.6782 & 0.1586\\
LLaMA 3 8B Instruct                & 0.4022 & 0.2978 & 0.6509 & 0.9427 & 0.1544\\
GPT-4 Turbo                        & 0.2825 & 0.1195 & 0.6186 & 0.7249 & 0.1556\\
GPT-4                              & 0.2906 & 0.1511 & 0.6619 & 0.7840 & 0.1570\\
GPT-3.5 Turbo                      & 0.3041 & 0.1334 & 0.6115 & 0.7027 & 0.1698\\
\textbf{MBPP (Mean)}               & 0.3009 & 0.1203 & 0.5621 & 0.6857 & 0.1866\\
\midrule
\textbf{HumanEval}                 &        &        &        &        &        \\
GPT-4o                             & 0.0927 & 0.0483 & 0.7181 & 0.7042 & 0.0460\\
Claude 4 Opus (2025/05/14)         & 0.0632 & 0.0067 & 0.4041 & 0.3000 & 0.0601\\
Gemini 2.5 Pro (preview 05/06)     & 0.0763 & 0.0295 & 0.7171 & 0.7188 & 0.0417\\
LLaMA 4 Maverick 17B               & 0.0876 & 0.0367 & 0.6242 & 0.5397 & 0.0583\\
Claude 4 Sonnet                    & 0.0605 & 0.0076 & 0.4098 & 0.2778 & 0.0533\\
DeepSeek-V3 (Mar 2024)             & 0.0781 & 0.0235 & 0.5346 & 0.4576 & 0.0552\\
GPT-o4 Mini                        & 0.0893 & 0.0388 & 0.7092 & 0.7051 & 0.0420\\
Mistral 8B                         & 0.1585 & 0.0911 & 0.7282 & 0.7381 & 0.0737\\
LLaMA 3 70B Instruct               & 0.1223 & 0.0928 & 0.8173 & 0.8395 & 0.0314\\
DeepSeek-Coder R1                  & 0.0724 & 0.0356 & 0.7994 & 0.7538 & 0.0180\\
Gemini 2.5 Flash (preview 05/20)   & 0.0796 & 0.0347 & 0.7584 & 0.7538 & 0.0339\\
Gemini 2.0 Flash Lite              & 0.1165 & 0.0343 & 0.4829 & 0.4118 & 0.0851\\
LLaMA 3 8B Instruct                & 0.2140 & 0.1954 & 0.8890 & 0.9245 & 0.0261\\
GPT-4 Turbo                        & 0.0912 & 0.0477 & 0.7528 & 0.7606 & 0.0398\\
GPT-4                              & 0.1144 & 0.0823 & 0.8127 & 0.8471 & 0.0424\\
GPT-3.5 Turbo                      & 0.1635 & 0.0906 & 0.8192 & 0.8537 & 0.0458\\
\textbf{HumanEval (Mean)}          & 0.1050 & 0.0560 & 0.6861 & 0.6616 & 0.0471\\
\bottomrule
\end{tabular}
\label{tab:rq1_full}
\end{table*}
\begin{table}[t]
    \centering
    \caption{Interpretation of Spearman's $\rho$, based on thresholds from \citet{schober2018correlation}.}
    \label{tab:spearman-interpretation}
    \begin{tabular}{c|l}
        \toprule
        \textbf{Spearman's $\rho$} & \textbf{Interpretation} \\
        \midrule
        0.00 -- 0.09 & Negligible correlation \\
        0.10 -- 0.39 & Weak correlation \\
        0.40 -- 0.69 & Moderate correlation \\
        0.70 -- 0.89 & Strong correlation \\
        0.90 -- 1.00 & Very strong correlation \\
        \bottomrule
    \end{tabular}
\end{table}
\begin{table}
    \caption{List of basic type-aware mutations over input $x$.}
    \label{tab:mutation_types}
    \centering
    \begin{tabular}{ll|ll}
    \toprule
    \textbf{Object Type} & \textbf{Mutation} & \textbf{Object Type} & \textbf{Mutation} \\
    \midrule
    \texttt{int}, \texttt{float} & Add  ($\pm$1, $\pm$10, random) 
        & \texttt{NoneType} & \texttt{None} \\
    \midrule
    \texttt{bool} & Random boolean (True or False)
        & \texttt{user-defined} & Shallow copy + mutate fields recursively \\
    \midrule
    \multirow{6}{*}{\texttt{str}}  
        & Insert char at random index & \multirow{5}{*}{\texttt{list}, \texttt{set}} & Insert random element at random index/key \\
    & Delete char at random index & & Insert dummy element at random index/key \\
    & Replace char w/ random ASCII  & & Swap two elements \\
    & Truncate random substring &\texttt{dict}, \texttt{tuple} & Duplicate entries randomly \\
    & Extend random substring &  & Insert entry at random index/key \\
    & Duplicate random substring & & Delete element at random index/key \\
    \bottomrule
    \end{tabular}
\end{table}

\end{document}